\documentclass[journa]{IEEEtran}
\usepackage{indentfirst}
\usepackage{graphicx}
\usepackage{float}
\usepackage{listings}
\usepackage{algorithmic,algorithm}
\usepackage{amsmath}
\usepackage{subfigure}
\usepackage{color}
\usepackage{dsfont}
\usepackage{bbm}
\usepackage{amssymb}
\usepackage{amsthm}
\usepackage{dsfont}

\newtheorem{lemma}{Lemma}
\newtheorem{theorem}{Theorem}
\newtheorem{corollary}{Corollary}
\usepackage{caption}
\usepackage{mathtools}
\usepackage{cuted}
\usepackage[keeplastbox]{flushend}
\usepackage{etoolbox}
\usepackage[utf8x]{inputenc}
\usepackage[T1]{fontenc}

\usepackage{comment}

\newcounter{tempEquationCounter} 
\newcounter{thisEquationNumber}
\newenvironment{floatEq}
{\setcounter{thisEquationNumber}{\value{equation}}\addtocounter{equation}{1}
\begin{figure*}[!t]
\normalsize\setcounter{tempEquationCounter}{\value{equation}}
\setcounter{equation}{\value{thisEquationNumber}}
}
{\setcounter{equation}{\value{tempEquationCounter}}
\hrulefill\vspace*{4pt}
\end{figure*}

}

\def\bE{{\mathbf{E}}}

\ifCLASSINFOpdf
\else
\fi
\begin{document}
%
\title{Spectral Efficiency Analysis of the Decoupled Access for Downlink
and Uplink in Two Tier Network}
%
%
%

\author{Zeeshan~Sattar,~\IEEEmembership{Student member,~IEEE,}
        Joao~V.C.~Evangelista,~\IEEEmembership{Student member,~IEEE,}
        ~Georges~Kaddoum,~\IEEEmembership{Member,~IEEE,}
        and~Na{\"\i}m~Batani,~\IEEEmembership{Member,~IEEE}
\thanks{
Z. Sattar, J.V.C. Evangelista, G. Kaddoum,  and N. Batani are with the Department
of Electrical Engineering, \'{E}cole de Technologie sup\'{e}rieure, Montr\'eal,
QC, H3C 1K3 CA, e-mail: \{zeeshan.sattar.1, joao-victor.de-carvalho-evangelista.1 \}@ens.etsmtl.ca and \{georges.kaddoum, naim.batani\} @etsmtl.ca}}

\maketitle

\begin{abstract}
This paper analyzes the efficacy of decoupled wireless access in a two-tier heterogeneous network. The decoupled wireless access and its performance benefits have been studied in different scenarios recently. In this paper, an in-depth analysis on its efficacy from spectral efficiency perspective is provided. To achieve this task, (i) new closed form expressions for probability of association of user equipment with different tiers employing different frequency bands (i.e., microwave and millimeter wave) with different pathloss exponents are derived using univariate Fox's H-functions; (ii) Distributions of the distance to the serving base stations are also derived; (iii) Exact expressions of spectral efficiency for different association cases are further obtained using bivariate Fox's H-functions. Furthermore, rigorous simulation results are provided which validate the aforementioned analytical results. In addition to that, a detailed discussion on the decoupling gain of decoupled wireless access and its efficacy is also provided. Lastly, despite the improvement provided by the decoupled wireless access, which is evident from the results presented in this paper, few questions are raised on its pragmatic value.
\end{abstract}


%
\IEEEpeerreviewmaketitle


\section{Introduction}
\label{introduction}

\IEEEPARstart{W}{hat} shape the next generation of communication systems will take is a question which can not be answered in one line or in other words there can never be just one answer to this question. The evolution of technology and recent advancement in the available computing power gave us plenty of room to think out of the box. Therefore, when the research community brought fifth generation (5G) of communication systems on the table, lots of innovative ideas came into existence \cite{ge20165g}. Few of them will definitely see the light of the practical world and many of them will get lost somewhere inside the research laboratories of academia only to be found again on a later date for another generation of communication systems. The main quest of 5G is to provide seamless coverage, hot-spot high capacity, low end to end latency, and massive connections \cite{andrews2014will, yang2018dense}. To meet these requirements two of the most promising candidates are the network densification and the use of extremely high frequencies (EHF) which are commonly known as millimeter wave (mmW) band. Here, the network densification refers to the paradigm shift of single-tier homogeneous cellular networks towards multi-tier heterogeneous cellular networks (HetNets) \cite{andrews2012femtocells}. In HetNets, different tiers of base stations (BSs) typically use different transmit powers which result in significantly different interference levels \cite{andrews2013seven}.
Therefore whether the conventional way of cell association i.e., coupled access, where a user connects to a single BS for both uplink and downlink transmission  would be optimal in HetNets came under scrutiny of the research community. 

Since the inception of cellular communication systems, coupled access is the only way for any   user equipment (UE) to connect to a BS. This conventional way of association to a BS recently has been 
challenged in form of decoupled wireless access \cite{elshaer2014downlink,boccardi2016decouple,elshaer2016downlink}. The concept of decoupled wireless access argues on the optimality of choosing the same BS for both uplink and downlink transmissions, and proposes to give the liberty to UEs to simultaneously connect to  two different BSs from any two different tiers of BSs for uplink and downlink transmissions. Though, both intuition and probability theory supports this idea for a simple fact that if we increase the size of the set of BSs to choose from, it would definitely result in a better performance in terms of coverage and spectral efficiency. In addition, the decoupled wireless access also breaks the channel reciprocity by its very design, so indirectly it also raises questions on the way we typically estimate the channel. Therefore, despite the potential benefits of the decoupled wireless access as it is evident in theory, its pragmatic value is still in question \cite{zsatar2017pimrc}. This inspires us to scrutinize the potential benefits of the decoupled wireless access and compare it to the conventional coupled access. To make things mathematically simple and tractable yet robust, we didn't employ any blockage model but used different pathloss exponents and transmit powers to emulate the characteristics of two different kinds of BSs (i.e., microwave BSs and mmW BSs) operating in different tiers.

The rationale behind omitting the blockage model is depicted in \cite{zsatar2017pimrc} which is an early study for this work, where a very practical blockage model was used as described in \cite{gapeyenko2016analysis}. We reached to the conclusion that though it does add the pragmatic value to the analytical model, it does not significantly affect the analysis of the average spectral efficiency. More detailed discussion on this assumption is given in section \ref{propagation_assumptions}.

\subsection{Related Work}
The idea of multi-tier cellular system is not new and it has been under the lens of both academia and industry for a long time. For example, authors in \cite{chandrasekhar2009coverage} investigated the case of two-tier cellular systems with universal frequency reuse. They study the case of single user (SU)
and multi-user (MU) multiple antenna methods to mitigate cross-tier interference and the ‘near-far’ deadspot coverage in a two tier network. They further provided location-assisted power control scheme for regulating femtocell BSs transmit powers.

Even though a lot of work has been done on the performance analysis of cellular systems, the scientific community is still working on new theoretical tools to evaluate the network performance of cellular systems. For example, authors in \cite{ge2015spatial} proposed new spatial spectrum and energy efficiency models for
Poisson-Voronoi tessellation (PVT) random cellular networks. In \cite{di2015stochastic}, the author provided a very detailed mathematical framework based on stochastic geometry to model multi-tier millimeter wave cellular network. An exact analytical model to derive coverage probability and average rate in form of numerical integrations are derived. Furthermore, to provide results in closed-form, approximated analytical models are also derived. In addition to that, a detailed discussion on the noise-limited approximation for typical millimeter wave network deployments is also provided.

Recently there has been significant amount of progress on the analysis of the decoupled wireless access  \cite{zhang2017uplink,li2018uplink,shi2018decoupled,aravanis2018closed}. Right after its proposal in \cite{elshaer2014downlink}, the first analytical analysis of the decoupled wireless access has been done by Smiljkovikj \textit{et al.} in \cite{smiljkovikj2015analysis, smiljkovikj2015efficiency}. 

In \cite{smiljkovikj2015analysis}, Smiljkovikj \textit{et al.} analyzed a two-tier network with macro cells and small cells. They provided analytical expressions for probability of associations of UEs to different tiers and average throughput of UEs associated to different tiers. In \cite{smiljkovikj2015efficiency}, the authors provided a deeper analysis  on the benefits of the decoupled wireless access by analyzing its spectral and energy efficiencies. 

In a more recent work \cite{zhang2017uplink}, the authors provided a comparative analysis of the decoupled and the coupled wireless access for two kinds of UE's distributions, namely uniform and clustered distributions, which are modeled as Poisson point and Neyman–Scott cluster processes, respectively. They borrowed analytical expressions of probability of cell association and distance distribution of a UE and its serving BS from \cite{jo2012heterogeneous} and derived new analytical expressions for average user rate for two UE's distributions.

In \cite{li2018uplink}, the authors analyzed the decoupled wireless access in multiuser multiple-input multiple-output (MIMO) HetNet scenario. They derived cell association probabilities with respect to the load balancing in BSs. They also compared the decoupled and coupled wireless access based on the load balancing in BSs and provided new analytical expressions for uplink spectral efficiency. Moreover, they also derived the lower bounds on the uplink spectral efficiency where interference is shown to be suppressed by multiple antennas at BSs.

In \cite{shi2018decoupled}, the authors provided a theoretical work on the impact of decoupled access in multi-tier HetNet. Using tools from stochastic geometry they derived general expression of association probability to a particular tier of BSs. Furthermore, a detailed analytical work on the impact of decoupled access on the coverage probability is also provided. 

In \cite{aravanis2018closed}, the authors derived analytical bounds in closed form for the uplink ergodic capacity as a function of the density of BSs of different tiers for the decoupled access scenario. The novelty of their work is to accommodate the backbone network congestion and the synchronization of the acknowledgments of the decoupled channels into their analytical expression.

A semi-analytical analysis of the decoupled wireless access is provided in \cite{zsatar2017pimrc}. The decoupled and coupled wireless access are compared for two-tier network employing a realistic blockage  model proposed in \cite{gapeyenko2016analysis}, where human body is considered as a blockage to a tier of BSs operating on mmW frequencies. Despite the practical nature of the blockage model used, authors came to this conclusion that it made the analytical analysis intractable. Therefore authors in \cite{elshaer2016downlink} used a rather simple blockage model proposed in \cite{singh2015tractable} to develop a general analytical model to
characterize and derive the uplink and downlink cell associations. Even for that simple blockage model, there analytical expressions for the cell association are too complicated for further analysis. For example, to study the distance distribution of a UE to its serving BS and to derive the expressions for spectral efficiency based on the analytical model in \cite{elshaer2016downlink}, the only option is to solve a plethora of nested numerical integrations.

\subsection{Novelty and Contributions}

The main goal of this paper is to investigate the spectral efficiency gain of the decoupled wireless access over its coupled counterpart. Additionally it provides more compact and robust analytical model. The novelty of our analytical analysis is the fact that it accommodates variable transmit powers and different pathloss exponents for different tiers in its formulation. Moreover, instead of a plethora of numerical integrations, Fox's H-function and its integration properties are used to solve numerical integrations into compact form. We believe that the proposed analytical model provides valuable insights into the mathematical analysis of the problem under consideration and several other related work.  Furthermore, one can easily compute the first and second order derivatives of Fox's H-function, which are often used in numerical optimization methods \cite{kilbas2004h}. The insights obtained from the outcome of this paper, regarding the solution of complex numerical integrations into Fox's H-function form would inspire researchers to provide compact analytical models. The key contributions of this paper are listed as follows. 
\begin{itemize}
\item New closed form expressions of joint probability of uplink and downlink cell associations are derived for a two-tier network. 

\item Univariate Fox's H-function is used in our analytical closed form expression to accommodate different pathloss exponents and variable transmit powers for different tiers (i.e., conventional mircowave BSs and mmW BSs).

\item  The distance distributions of a UE to its serving BSs are also derived for three  possible cases of the cell associations.

\item Finally exact  expressions of spectral efficiencies for three possible cases of cell association are derived using bivariate Fox's H-function. The motivation to use Fox's H-function is to formulate analytical expressions in compact and modular form\footnote{The univariate Fox's H-function is implemented for Mathematica in \cite{yilmaz2009product, ansari2013sum}, and for MATLAB in \cite{peppas2012simple}, whereas the implementation of the bivariate Fox's H-function is given for MATLAB in \cite{peppas2012new}.}. In addition to that, three of the most important features of this function are: (i) many of the special functions in similar category of mathematical analysis are its special cases; (ii) the integral of the product of two H-functions is again a Fox's H-function \cite{mathai2009h,kilbas2004h}; (iii) many generalized channel models can be formulated in a compact form using H-functions \cite{long2018,ICCWORKSHOPlong2018, 7108024}.

\end{itemize}

\subsection{Organization}
The rest of the paper is organized as follows.

In Section \ref{system_model}, the system model is described in detail, which includes the propagation assumptions and cell association criteria. In Section \ref{sec:analytical_analysis}, the joint association probabilities, distance distributions of a typical UE and its serving BSs are derived. Furthermore, analytical expressions for average user rates and spectral efficiencies are also provided in this section.  In Section \ref{sec:num_sim_results}, discussion on the obtained numerical and simulation results is provided and Section \ref{sec:conclusions} concludes the paper.

\section{System Model}
\label{system_model}

\begin{figure}[t]\centering
\includegraphics[width=\columnwidth]{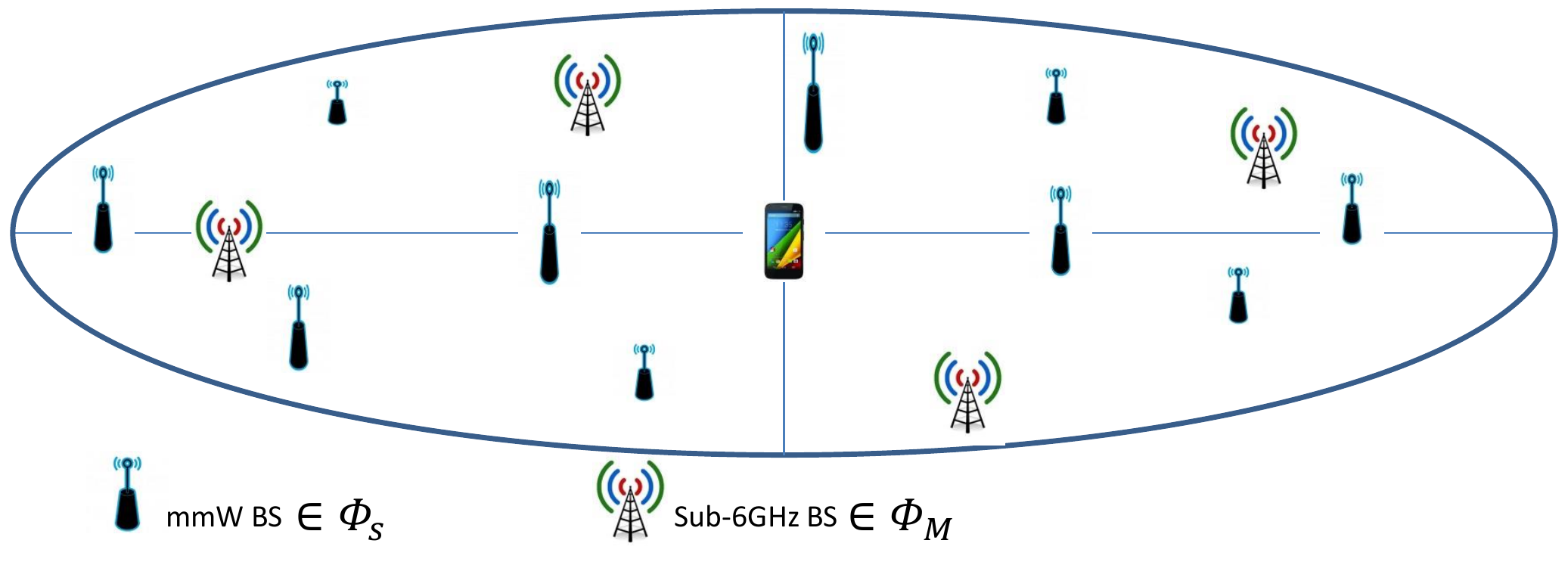}
\caption{System Model}
\label{fig:SystemModel}
\end{figure}

We consider a two-tier HetNet, where sub-6GHz (i.e., conventional microwave or $Mcell$) BSs and mmW (i.e., $Scell$) BSs are modeled using independent homogeneous Poisson point process (PPP) as shown in Fig. \ref{fig:SystemModel}. All the BSs are uniformly distributed in a circular area with radius $\mu$ . We use ${\Phi _k}$ to denote the set of points obtained through PPP with density ${\lambda _k}$, that can be explicitly written as
\begin{eqnarray}
\label{eq:definition}
{\Phi _k} \buildrel \Delta \over = \{ {x_{k,i}} \in {\mathbb{R}^2}:i \in {\mathbb{N}_ + }\} ,
\nonumber
\end{eqnarray}
where the index $k \in \{ M, S\}$)  for $Mcell$ and $Scell$ BSs, respectively. Moreover, all the UEs are also assumed to form an independent PPP with density ${\lambda _u}$ and they are denoted by a set ${\Phi _u}$ given as

\begin{eqnarray}
\label{eq:definition2}
{\Phi _u} \buildrel \Delta \over = \{ u_j \in {\mathbb{R}^2}:j \in {\mathbb{N}_ + }\} .
\nonumber
\end{eqnarray}

The transmit powers for downlink and uplink transmissions are $P_M$ and $Q_M$, respectively, for the tier of $Mcell$ BSs. Similarly, for the tier of $Scell$ BSs, $P_S$ and $Q_S$, respectively, are the transmit powers for downlink and uplink transmissions.  Besides that, since the distribution of a point process is completely indifferent to the addition of a node at the origin, allowed by Slivnyak's theorem \cite{chiu2013stochastic}, the analysis is done for a typical UE located at $u_j=(0,0)$.

\subsection{Propagation Assumptions}
\label{propagation_assumptions}
In this subsection all the major assumptions critical to the analytical analysis are listed.
\begin{itemize}
\item  { As mentioned in section \ref{introduction}, the rationale behind the assumption about blockage model is that for the higher values of $\lambda_S/\lambda_M$ the probability of association curves takes the same shape with and without any blockage model \cite{zsatar2017pimrc}. In addition to that, even for the smaller values of $\lambda_S/\lambda_M$, the probability of association curves shows that the inclusion of blockage model decreases the probability of the case where UEs choose to decouple \cite{zsatar2017pimrc}. Hence, the inclusion of blockage model will not significantly affect our results in anyway. Therefore, in the literature, the use of blockage model has been done for the cases when either the study is primarily about the channel modeling or just probability of associations.}

\item {
With regards to the effect of shadowing in our analysis, it should be mentioned that the authors in \cite{andrews2011tractable} pointed out that even a simplified model that only considers Rayleigh fading can closely track an actual base station deployment with lognormal shadowing. Therefore as far as the sub-6GHz networks are concerned, the randomness of the PPP BS locations emulates the shadowing effect, hence shadowing is ignored in the sub-6GHz model. On the other hand in mmW networks, a blockage model introduces similar effect to shadowing. Since in our work we omitted blockage model, and explained in detail the rationale behind this assumption, shadowing is ignored for mmW networks too. Moreover, to further demonstrate the negligible effect of shadowing in decoupled wireless access, a detailed discussion along with related results are provided in section \ref{sec:num_sim_results}.   
}

\item It is assumed that $P_M>P_S$, as $Mcell$ BSs suppose to have more transmit power to provide coverage to all the UE in its cell. It is also assumed that $Q_M{\geq}Q_S$, which is  not only an intuitive assumption as UEs need more power to transmit to far BSs but also based on the difference between the maximum allowed transmit power for mmW and sub-6GHz UEs \cite{colombi2015implications}.

\item In our system model beamforming gains from massive array of antenna elements are only accounted for $Scell$ BSs.  Though, even in sub-6GHz domain, antenna pattern has certain shape, it is assumed that all UEs and $Mcell$ BSs have omni directional antennas. The rationale behind this assumption is that in hybrid BSs' deployment, the $Mcell$ BSs will provide an umbrella coverage to all UEs, on the other hand  $Scell$ BSs will mainly focus on the high capacity links with individual UEs.

\item It is assumed that in both uplink and downlink, a typical UE associates with a BS based on the received power.

\item Another assumption that mmW networks are noise-limited and sub-6GHz networks are interference limited is also considered. { Additionally, from the perspective  of medium access control (MAC), authors in \cite{singh2011interference} provides a detailed  interference analysis, which shows that highly directional links can indeed be modeled as pseudowired. On the other hand, authors in \cite{shokri2016transitional} discussed that mmW networks may exhibit non-negligible transitional behavior from a noise-limited regime to an interference-limited. }{ The practical aspect of the noise-limited mmW networks is motivated by the work in \cite{di2015stochastic, andrews2017modeling} where the authors discuss in detail the pragmatic value of this assumption which also simplifies the mathematical analysis. Furthermore, the authors in \cite{rangan2014millimeter, akdeniz2014millimeter} did simulation based studies on a measurement-based mmW channel model.  It was observed that the impact of thermal noise on coverage dominates that of out of cell interference in mmW networks.}

The case of noise-limited mmW networks has been considered and motivated in \cite{singh2015tractable} and later validated by \cite{elshaer2016downlink} even for high densities of mmW $Scells$. Moreover, due to the  orthogonality of both sub-6GHz and mmW networks, no interference is assumed between the two tiers.
\end{itemize}

\subsection{Cell Association Criteria}

The typical UE associates with a BS in uplink at $x*\in \Phi_l$, where $l \in \{S, M\}$ if and only if
\begin{eqnarray}
\label{eq:ul_association}
{Q_{l}}G_{l}||x*|{|^{ - {\alpha _l}}} &\ge& {Q_{k}}G_{k}||x_{k,i}|{|^{ - {\alpha _k}}},\cr 
&&\quad \forall k \in \{ S, M\}.
\end{eqnarray}

Similarly, a typical UE associates with a BS in downlink at $x* \in \Phi_l$ if and only if

\begin{eqnarray}
\label{eq:dl_association}
{P_{l}}G_{l}||x*|{|^{ - {\alpha _l}}} &\ge& {P_{k}}G_{k}||x_{k,i}|{|^{ - {\alpha _k}}},\cr 
&&\quad \forall k \in \{ S, M\},
\end{eqnarray}
where $G_k$ and $\alpha_k$ are antenna gain and path loss exponent for the communication link UE-$kcell$ BS, respectively. 

Moreover, based on the assumption of noise-limited mmW network and  interference limited sub-6Ghz networks, the uplink/downlink mmW signal-to-noise-ratios (SNRs) and sub-6GHz signal-to-interference and noise ratios (SINRs) take the following form

\begin{eqnarray}
\label{eq:snirs} \nonumber
{\rm{SIN}}{{\rm{R}}_{{\rm{UL}},M}} & =& \frac{{{{\rm{Q}}_M}{G_M}{h_{0,x*}}||x*|{|^{ - {\alpha _M}}}}}{{{I_{{\rm{UL}},M}} + \sigma _M^2}}\\ \nonumber
{\rm{SIN}}{{\rm{R}}_{{\rm{DL}},M}} &=& \frac{{{{\rm{P}}_M}{G_M}{h_{x*,0}}||x*|{|^{ - {\alpha _M}}}}}{{{I_{{\rm{DL}},M}} + \sigma _M^2}}\\ \nonumber
{\rm{SN}}{{\rm{R}}_{{\rm{UL}},S}} &=& \frac{{{{\rm{Q}}_S}{G_S}{h_{0,x*}}||x*|{|^{ - {\alpha _S}}}}}{{\sigma _S^2}}\\ 
{\rm{SN}}{{\rm{R}}_{{\rm{DL}},S}} &=& \frac{{{{\rm{P}}_S}{G_S}{h_{x*,0}}||x*|{|^{ - {\alpha _S}}}}}{{\sigma _S^2}},
\end{eqnarray}
where $\sigma_k^2$ and $h$ are the noise variance for the communication link UE-$kcell$ BS and  small scale fading power gain, respectively. For the rest of the paper we denote $\bar P_k=P_k G_k$ and $\bar Q_k=Q_k G_k$.


\section{Analytical Analysis}
\label{sec:analytical_analysis}
In the quest of deriving the analytical expressions for the spectral efficiency in different association scenarios, the first step is to derive the analytical expressions for the joint probability of association. The second step is to derive the distance distributions to the serving BSs for different association scenarios.

\subsection{Joint Association Probabilities}
\label{subsec:joint_association_probabilities}
Considering the model proposed in Section~\ref{system_model}, if a typical UE has a liberty to choose at most two different BSs for uplink and downlink transmissions from two different tier of BSs, then the association process can lead to one of the following four cases:
\begin{itemize}
\item Case 1: Uplink BS = Downlink BS = $Mcell$ BS
\item Case 2: Uplink BS = $Scell$ BS, Downlink BS = $Mcell$ BS
\item Case 3: Uplink BS = $Mcell$ BS, Downlink BS = $Scell$ BS
\item Case 4: Uplink BS = Downlink BS = $Scell$ BS
\end{itemize}
The joint association probabilities for homogeneous (i.e., all UEs communicate with the same transmit power) and heterogeneous user domain (i.e., UEs vary their transmit power levels with respect to the BS's tier they are connected to ) are already elaborated in \cite{smiljkovikj2015analysis} and \cite{smiljkovikj2015efficiency}, respectively. Since, in the two-tier network under consideration, both tiers which are operating on significantly different frequency bands (sub-6GHz and mmW) possess drastically different propagation characteristics.  Hence, the novelty of this work  is to accommodate different pathloss exponents in the closed form expressions of joint association probabilities, which further leads to robust expressions for spectral efficiency of the association cases under consideration.  

Let $\{X_k\}_{k\in \{M, S\}}$ denotes the distance from the nearest BS in the $k^{th}$tier to the typical UE located at $u_j=(0,0)$. We can derive the probability density function (pdf) of $X_k$ and cumulative distribution function (cdf) by the null probability of a 2D PPP \cite{chiu2013stochastic} as follows:

\begin{eqnarray}
\label{eq:nearest_distance_BS}
{f_{{X_k}}}(x) &=& 2\pi {\lambda _k}x\exp ( - \pi {\lambda _k}{x^2}), \quad x \ge 0,\\
\label{eq:nearest_distance_BS_2}
{F_{{X_k}}}(x) &=& 1 - \exp ( - \pi {\lambda _k}{x^2}),\quad x \ge 0.
\end{eqnarray}

Based on the cell association rules given by  (\ref{eq:ul_association}) and (\ref{eq:dl_association}) we can derive the joint cell association probabilities for the four cases under consideration as follows:

\subsubsection{Case 1: Uplink BS = Downlink BS = $Mcell$ BS} The probability that a UE associates to $Mcell$ BS for both uplink and downlink transmissions is given by

\begin{eqnarray}
\label{eq:pr_association_case1}
\Pr ({\rm{Case}}1)&=&\Pr \left(X_M^{ - {\alpha _M}} > \frac{{{{\bar P}_S}}}{{{{\bar P}_M}}}X_S^{ - {\alpha _S}};\right. \cr 
&& 
\quad  \quad  \quad  \quad \left. X_M^{ - {\alpha _M}} > \frac{{{{\bar Q}_S}}}{{{{\bar Q}_M}}}X_S^{ - {\alpha _S}}\right).
\end{eqnarray}

Based on the discussion on different power levels, $\frac{{{{\bar Q}_S}}}{{{{\bar Q}_M}}} > \frac{{{{\bar P}_S}}}{{{{\bar P}_M}}}$, therefore the joint probability reduces to the following form
\begin{eqnarray}
\label{eq:pr_association_case1_1}
\Pr ({\rm{Case}}1)&=& \left(X_M^{ - {\alpha _M}} > \frac{{{{\bar Q}_S}}}{{{{\bar Q}_M}}}X_S^{ - {\alpha _S}}\right) \cr
&=& \left({X_M} < {\Biggl( {\frac{{{{\bar Q}_M}}}{{{{\bar Q}_S}}}} \Biggr)^{\frac{1}{{{\alpha _M}}}}}X_S^{  \frac{{{\alpha _S}}}{{{\alpha _M}}}}\right).
\end{eqnarray}

\begin{lemma}
\label{lemma:1}
The joint probability of association of a typical UE for the association case 1 in closed form can be formulated as

\begin{eqnarray}
\label{eq:pr_association_case1_H}
\Pr ({\rm{Case}}1) &=& 1-{\frac{1}{2}}{\frac{\alpha_M}{\alpha_S}}{\rm{H}}_{1,1}^{1,1}\Biggl[ {{z_{1}}\Biggl| {\begin{array}{*{20}{c}}
{(0,\frac{1}{2})}\\
{(0,\frac{1}{2}\frac{{{\alpha _M}}}{{{\alpha _S}}})}
\end{array}} \Biggr.} \Biggr],
\end{eqnarray}
where $z_{1} = {\Biggl( {\sqrt {\pi {\lambda _M}} } \Biggr)^{\frac{{{\alpha _M}}}{{{\alpha _S}}}}}{\Biggl( {\frac{{{{\bar Q}_M}}}{{{{\bar Q}_S}}}} \Biggr)^{\frac{1}{{{\alpha _S}}}}}{\Biggl( {\sqrt {\pi {\lambda _S}} } \Biggr)^{ - 1}}$.
\end{lemma}

\begin{proof}
The development and proof are shown in Appendix \ref{Lemma1_3}.
\end{proof}


\subsubsection{Case 2: Uplink BS = $Scell$ BS, Downlink BS = $Mcell$ BS} The probability that a UE associates to $Scell$ BS for uplink and $Mcell$ BS for downlink transmission is given by

\begin{eqnarray}
\label{eq:pr_association_case2}
\Pr ({\rm{Case}}2) &=& \Pr \left( {X_M^{ - {\alpha _M}} > \frac{{{{\bar P}_S}}}{{{{\bar P}_M}}}X_S^{ - {\alpha _S}};} \right. \cr
&& \quad \quad \quad \quad {\rm{ }}\left. {X_M^{ - {\alpha _M}} \le \frac{{{{\bar Q}_S}}}{{{{\bar Q}_M}}}X_S^{ - {\alpha _S}}} \right).
\end{eqnarray}

Since, (\ref{eq:pr_association_case2}) is defined over an intersection of two regions, it can be written as 

\begin{eqnarray}
\label{eq:pr_associaion_case2_1}
\Pr ({\rm{Case}}2) &=& \Pr \left( {X_M^{ - {\alpha _M}} > \frac{{{{\bar P}_S}}}{{{{\bar P}_M}}}X_S^{ - {\alpha _S}}} \right) \cr 
&& -\Pr \left( {X_M^{ - {\alpha _M}} > \frac{{{{\bar Q}_S}}}{{{{\bar Q}_M}}}X_S^{ - {\alpha _S}}} \right)
\end{eqnarray}

\begin{lemma}
\label{lemma_2}
The joint probability of association of a typical UE for the association case 2 in closed form can be formulated as

\begin{eqnarray}
\label{eq;pr_association_case2_H}
\Pr (\rm{Case}2)&=& {\frac{1}{2}}{\frac{\alpha_M}{\alpha_S}}\left( H_{1,1}^{1,1}\left[ {{z_1}\left| {\begin{array}{*{20}{c}}
{(0,\frac{1}{2})}\\
{(0,\frac{1}{2}\frac{{{\alpha _M}}}{{{\alpha _S}}})}
\end{array}} \right.} \right] \right. \cr
&& \left.- H_{1,1}^{1,1}\left[ {{z_2}\left| {\begin{array}{*{20}{c}}
{(0,\frac{1}{2})}\\
{(0,\frac{1}{2}\frac{{{\alpha _M}}}{{{\alpha _S}}})}
\end{array}} \right.} \right]\right),
\end{eqnarray}
where $z_{2} = {\Biggl( {\sqrt {\pi {\lambda _M}} } \Biggr)^{\frac{{{\alpha _M}}}{{{\alpha _S}}}}}{\Biggl( {\frac{{{{\bar P}_M}}}{{{{\bar P}_S}}}} \Biggr)^{\frac{1}{{{\alpha _S}}}}}{\Biggl( {\sqrt {\pi {\lambda _S}} } \Biggr)^{ - 1}}$.
\end{lemma}
\begin{proof}
Similar to the proof of Lemma \ref{lemma:1} provided in Appendix \ref{Lemma1_3}.
\end{proof}


\subsubsection{Case 3: Uplink BS = $Mcell$ BS, Downlink BS = $Scell$ BS} The probability that a UE associates to $Mcell$ BS for uplink and $Scell$ BS for downlink transmission is given by

\begin{eqnarray}
\label{eq:pr_association_case3}
\Pr ({\rm{Case}}3) &=& \Pr \left( {X_M^{ - {\alpha _M}} \leq \frac{{{{\bar P}_S}}}{{{{\bar P}_M}}}X_S^{ - {\alpha _S}};} \right. \cr
&& \quad \quad \quad \quad {\rm{ }}\left. {X_M^{ - {\alpha _M}} > \frac{{{{\bar Q}_S}}}{{{{\bar Q}_M}}}X_S^{ - {\alpha _S}}} \right).
\end{eqnarray}

Since, there is no region which satisfies the domain of joint probability in  (\ref{eq:pr_association_case3}), the $\Pr(\rm{Case}3)=0;$


\subsubsection{Case 4: Uplink BS = Downlink BS = $Scell$ BS} The probability that a UE associates to $Scell$ BS for both uplink and downlink transmissions is given by

\begin{eqnarray}
\label{eq:pr_association_case4}
\Pr ({\rm{Case}}4)&=&\Pr \left(X_S^{ - {\alpha _S}} \geq \frac{{{{\bar P}_M}}}{{{{\bar P}_S}}}X_M^{ - {\alpha _M}};\right. \cr 
&& 
\quad  \quad  \quad  \quad \left. X_S^{ - {\alpha _S}} > \frac{{{{\bar Q}_M}}}{{{{\bar Q}_S}}}X_M^{ - {\alpha _M}}\right).
\end{eqnarray}

Since, we are assuming that $\frac{{{{\bar Q}_S}}}{{{{\bar Q}_M}}} > \frac{{{{\bar P}_S}}}{{{{\bar P}_M}}}$, the joint probability of  (\ref{eq:pr_association_case4}) reduces to the following form

\begin{eqnarray}
\label{eq:pr_association_case4_1}
\Pr ({\rm{Case}}4)&=&\Pr \left(X_S^{ - {\alpha _S}} \geq \frac{{{{\bar P}_M}}}{{{{\bar P}_S}}}X_M^{ - {\alpha _M}}\right).
\end{eqnarray}

\begin{lemma}
\label{lemma_3}
The joint probability of association of a typical UE for the association case 4 in closed form can be formulated as

\begin{eqnarray}
\label{eq:pr_association_case1_H}
\Pr ({\rm{Case}}4) &=& {\frac{1}{2}}{\frac{\alpha_M}{\alpha_S}}{\rm{H}}_{1,1}^{1,1}\Biggl[ {{z_{2}}\Biggl| {\begin{array}{*{20}{c}}
{(0,\frac{1}{2})}\\
{(0,\frac{1}{2}\frac{{{\alpha _M}}}{{{\alpha _S}}})}
\end{array}} \Biggr.} \Biggr].
\end{eqnarray}
\end{lemma}
\begin{proof}
Similar to the proof of Lemma \ref{lemma:1} provided in Appendix \ref{Lemma1_3}.
\end{proof}


\subsection{Distance distributions of a typical UE to its Serving BSs}

In this sub-section we derive the distance distributions of a typical UE to its serving BSs for all the three cases discussed in section \ref{subsec:joint_association_probabilities}. It is important to emphasize here that the serving BS may not be the nearest one to UE. 

\begin{lemma}
\label{lemma:4}
The distance distribution of a typical UE to its serving BS for the association case 1 is formulated as

\begin{eqnarray}
\label{eq:distance_case1}
{f_{{X_M}|{\rm{Case}}1}} = \frac{{\left( {\exp \left( { - \pi {\lambda _S}{{\left( {\frac{{{\bar Q_S}}}{{{\bar Q_M}}}} \right)}^{\frac{2}{{{\alpha _S}}}}}{x^{\frac{{2{\alpha _M}}}{{{\alpha _S}}}}}} \right)} \right) \cdot {f_{{X_M}}}}}{{\Pr (\rm{Case}1)}}.
\end{eqnarray}

\end{lemma}
\begin{proof}
The development and proof are shown in Appendix \ref{Lemma4_6}
\end{proof}

\begin{lemma}

\begin{floatEq}
\begin{subequations}
\begin{eqnarray}
\label{eq:distance_case2_a}
{f_{{X_M}|{\rm{Case}}2}} &=& \frac{{\left( {\exp \left( { - \pi {\lambda _S}{{\left( {\frac{{{\bar P_S}}}{{{\bar P_M}}}} \right)}^{\frac{2}{{{\alpha _S}}}}}{x^{\frac{{2{\alpha _M}}}{{{\alpha _S}}}}}} \right) - \exp \left( { - \pi {\lambda _S}{{\left( {\frac{{{\bar Q_S}}}{{{\bar Q_M}}}} \right)}^{\frac{2}{{{\alpha _S}}}}}{x^{\frac{{2{\alpha _M}}}{{{\alpha _S}}}}}} \right)} \right) \cdot {f_{{X_M}}}}}{{\Pr (\rm{Case}2)}}\\
\label{eq:distance_case2_b}
{f_{{X_S}|{\rm{Case}}2}} &=& \frac{{\left( {\exp \left( { - \pi {\lambda _M}{{\left( {\frac{{{\bar P_M}}}{{{\bar P_S}}}} \right)}^{\frac{2}{{{\alpha _M}}}}}{x^{\frac{{2{\alpha _S}}}{{{\alpha _M}}}}}} \right) - \exp \left( { - \pi {\lambda _M}{{\left( {\frac{{{\bar Q_M}}}{{{\bar Q_S}}}} \right)}^{\frac{2}{{{\alpha _M}}}}}{x^{\frac{{2{\alpha _S}}}{{{\alpha _M}}}}}} \right)} \right) \cdot {f_{{X_S}}}}}{{\Pr (\rm{Case}2)}}
\end{eqnarray}
\end{subequations}
\end{floatEq}

\label{lemma:5}
The distance distributions of a typical UE to its serving BSs for the association case 2 is formulated as given in  (\ref{eq:distance_case2_a}) and (\ref{eq:distance_case2_b}).

\end{lemma}
\begin{proof}
Similar to the proof of Lemma \ref{lemma:4} provided in Appendix \ref{Lemma4_6}.
\end{proof}

\begin{lemma}
\label{lemma:6}
The distance distribution of a typical UE to its serving BS for the association case 4 is formulated as

\begin{eqnarray}
\label{eq:distance_case4}
{f_{{X_S}|{\rm{Case}}4}} = \frac{{\left( {\exp \left( { - \pi {\lambda _M}{{\left( {\frac{{{\bar P_M}}}{{{\bar P_S}}}} \right)}^{\frac{2}{{{\alpha _M}}}}}{x^{\frac{{2{\alpha _S}}}{{{\alpha _M}}}}}} \right)} \right) \cdot {f_{{X_S}}}}}{{\Pr (\rm{Case}4)}}.
\end{eqnarray}

\end{lemma}
\begin{proof}
Similar to the proof of Lemma \ref{lemma:4} provided in Appendix \ref{Lemma4_6}.
\end{proof}

\subsection{Spectral Efficiency}

In this sub-section we derive analytical expressions of spectral efficiency for each association case separately. Results from Lemma  \ref{lemma:1} to \ref{lemma:6} are used to achieve this task.

The average system spectral efficiency can be formulated as 

\begin{eqnarray}
\label{eq:average_spectral_efficiency}
{\rm{SE}}&=&\sum\limits_{i = 1}^4 {{\rm{SE}}({\rm{Case }}\,i)\Pr({\rm{Case }}\,i)} \cr
&=&\sum\limits_{i = 1}^4 {\frac{{{\tau _i}}}{W_k}\Pr({\rm{Case }}\,i))},
\end{eqnarray}
where $\tau_i$ is the average UE rate of association case $i$ and $W_k$ is the system bandwidth, here $k \in\{M,S\}$. First the average UE rate for each association case is derived separately, then the average system spectral efficiency is derived.

\begin{theorem}
\label{Th:Avg_user_rate_d}

Based on the distance distributions derived in Lemma \ref{lemma:4}, \ref{lemma:5}, \ref{lemma:6}, the average user rate for each  association case $i$ can be formulated as given in equations (\ref{eq:average_user_rate_case1_UL}) to (\ref{eq:average_user_rate_case4_DL}). The definitions of the variables therein are given in Table \ref{Table:definitions} and the structure of the bivariate Fox's H-function is defined in Appendix \ref{Appendix:definitions}.

\begin{eqnarray}
\label{eq:average_user_rate_case1_UL}
{\tau _{{\rm{UL}}{\rm{,Case1}}}} &=& \frac{{W_M2\pi {\lambda _M}{\beta _1}{\beta _2}{\beta _3}}}{{\Pr ({\rm{Case1}})}} \cr 
&& \cdot\int\limits_{t > 0} {\xi _{{\rm{UL}}{\rm{,}}M}^{ - 2}}  { {\rm{\hat H}}\left( {1{\rm{;}}\frac{{{\xi _6}}}{{{\xi _{{\rm{UL}}{\rm{,}}M}}}}{\rm{,}}\frac{{{\xi _5}}}{{{\xi _{{\rm{UL}}{\rm{,}}M}}}}} \right)} dt.
\end{eqnarray}
\begin{eqnarray}
\label{eq:average_user_rate_case1_DL}
{\tau _{{\rm{DL}}{\rm{,Case1}}}} &=& \frac{{W_M2\pi {\lambda _M}{\beta _1}{\beta _2}{\beta _3}}}{{\Pr ({\rm{Case1}})}} \cr 
&&\cdot \int\limits_{t > 0} {\xi _{{\rm{DL}}{\rm{,}}M}^{ - 2}}  { {\rm{\hat H}}\left( {1{\rm{;}}\frac{{{\xi _6}}}{{{\xi _{{\rm{DL}}{\rm{,}}S}}}}{\rm{,}}\frac{{{\xi _5}}}{{{\xi _{{\rm{DL}}{\rm{,}}S}}}}} \right)} dt.
\end{eqnarray}


\begin{floatEq}
\begin{eqnarray}
\label{eq:average_user_rate_case2_UL}
{\tau _{{\rm{UL}}{\rm{,Case2}}}} = \frac{{W_S2\pi {\lambda _S}{\beta _4}{\beta _5}{\beta _6}}}{{\Pr ({\rm{Case2}})}}\int\limits_{t > 0} {\xi _{{\rm{UL}}{\rm{,}}S}^{ - 2}} \left( {{\rm{\hat H}}\left( {4{\rm{;}}\frac{{{\xi _1}}}{{{\xi _{{\rm{UL}}{\rm{,}}S}}}}{\rm{,}}\frac{{{\xi _2}}}{{{\xi _{{\rm{UL}}{\rm{,}}S}}}}} \right) - {\rm{\hat H}}\left( {4{\rm{;}}\frac{{{\xi _3}}}{{{\xi _{{\rm{UL}}{\rm{,}}S}}}}{\rm{,}}\frac{{{\xi _2}}}{{{\xi _{{\rm{UL}}{\rm{,}}S}}}}} \right)} \right)dt.
\end{eqnarray}
\end{floatEq}


\begin{floatEq}
\begin{eqnarray}
\label{eq:average_user_rate_case2_DL}
{\tau _{{\rm{DL}}{\rm{,Case2}}}} = \frac{{W_M2\pi {\lambda _M}{\beta _1}{\beta _2}{\beta _3}}}{{\Pr ({\rm{Case2}})}}\int\limits_{t > 0} {\xi _{{\rm{DL}}{\rm{,}}M}^{ - 2}} \left( {{\rm{\hat H}}\left( {1{\rm{;}}\frac{{{\xi _4}}}{{{\xi _{{\rm{DL}}{\rm{,}}M}}}}{\rm{,}}\frac{{{\xi _5}}}{{{\xi _{{\rm{DL}}{\rm{,}}M}}}}} \right) - {\rm{\hat H}}\left( {1{\rm{;}}\frac{{{\xi _6}}}{{{\xi _{{\rm{DL}}{\rm{,}}M}}}}{\rm{,}}\frac{{{\xi _5}}}{{{\xi _{{\rm{DL}}{\rm{,}}M}}}}} \right)} \right)dt.
\end{eqnarray}
\end{floatEq}



\begin{eqnarray}
\label{eq:average_user_rate_case4_UL}
{\tau _{{\rm{UL}}{\rm{,Case4}}}} &=& \frac{{W_S2\pi {\lambda _S}{\beta _{4}}{\beta _{5}}{\beta _{6}}}}{{\Pr ({\rm{Case4}})}} \cr 
&& \cdot \int\limits_{t > 0} {\xi _{{\rm{UL}}{\rm{,}}S}^{ - 2}}  { {\rm{\hat H}}\left( {4{\rm{;}}\frac{{{\xi _3}}}{{{\xi _{{\rm{UL}}{\rm{,}}S}}}}{\rm{,}}\frac{{{\xi _2}}}{{{\xi _{{\rm{UL}}{\rm{,}}S}}}}} \right)} dt.
\end{eqnarray}


\begin{eqnarray}
\label{eq:average_user_rate_case4_DL}
{\tau _{{\rm{DL}}{\rm{,Case4}}}} &=& \frac{{W_S2\pi {\lambda _S}{\beta _{4}}{\beta _{5}}{\beta _{6}}}}{{\Pr ({\rm{Case4}})}} \cr 
&& \cdot \int\limits_{t > 0} {\xi _{{\rm{DL}}{\rm{,}}S}^{ - 2}}  { {\rm{\hat H}}\left( {4{\rm{;}}\frac{{{\xi _3}}}{{{\xi _{{\rm{DL}}{\rm{,}}S}}}}{\rm{,}}\frac{{{\xi _2}}}{{{\xi _{{\rm{DL}}{\rm{,}}S}}}}} \right)} dt.
\end{eqnarray}

\end{theorem}
\begin{proof}
The development and the proof are shown in Appendix \ref{appendix_avg_user_rate}.
\end{proof}



\begin{table*}
\centering
\caption{Definitions of variables/parameters used in the expressions of Fox's H-function.}
\label{Table:definitions}
\begin{tabular}{|l|l|l|l|}
\hline & & &
\\
Variables & Definitions  & Variables & Definitions           \\ & & &\\ \hline \hline & & &
\\
$\xi_{{\rm{UL}},M}$& $\frac{{{{\left( {\exp (t) - 1} \right)}^{\frac{1}{{{\alpha _M}}}}}}}{{{{\bar Q}_M}^{\frac{1}{{{\alpha _M}}}}}}$ & 

$\xi_{{\rm{DL}},M}$& $\frac{{{{\left( {\exp (t) - 1} \right)}^{\frac{1}{{{\alpha _M}}}}}}}{{{{\bar P}_M}^{\frac{1}{{{\alpha _M}}}}}}$         \\ & & &\\ \hline & & &

\\



$\xi_{{\rm{UL}},S}$&  $\frac{{{{\left( {\exp (t) - 1} \right)}^{\frac{1}{{{\alpha _S}}}}}}}{{{{\bar Q}_S}^{\frac{1}{{{\alpha _S}}}}}}$ & 

$\xi_{{\rm{DL}},S}$& $\frac{{{{\left( {\exp (t) - 1} \right)}^{\frac{1}{{{\alpha _S}}}}}}}{{{{\bar P}_S}^{\frac{1}{{{\alpha _S}}}}}}$ \\ & & & \\
\hline & & &
\\

$\xi_1$& ${\left( {\sqrt {\pi {\lambda _M}} } \right)^{\frac{{{\alpha _M}}}{{{\alpha _S}}}}}{\left( {\frac{{{\bar Q_M}}}{{{\bar Q_S}}}} \right)^{\frac{1}{{{\alpha _S}}}}} $ &

$\xi_2$&  $\sqrt {\pi {\lambda _S}} $\\ & & &\\
\hline & & &
\\
$\xi_3$& ${\left( {\sqrt {\pi {\lambda _M}} } \right)^{\frac{{{\alpha _M}}}{{{\alpha _S}}}}}{\left( {\frac{{{\bar P_M}}}{{{\bar P_S}}}} \right)^{\frac{1}{{{\alpha _S}}}}} $ & 

$\xi_4$&${\left( {\sqrt {\pi {\lambda _S}} } \right)^{\frac{{{\alpha _S}}}{{{\alpha _M}}}}}{\left( {\frac{{{\bar P_S}}}{{{\bar P_M}}}} \right)^{\frac{1}{{{\alpha _M}}}}} $ \\& & & \\
\hline & & &
\\
$\xi_5$& $\sqrt {\pi \left( {\lambda _M + \lambda _{\rm{IU}}G(t)}\right) } $ & 
$\xi_6$&${\left( {\sqrt {\pi {\lambda _S}} } \right)^{\frac{{{\alpha _S}}}{{{\alpha _M}}}}}{\left( {\frac{{{\bar Q_S}}}{{{\bar Q_M}}}} \right)^{\frac{1}{{{\alpha _M}}}}} $ \\ & & &\\
\hline & & &
\\
$\beta_1$& $\frac{1}{\alpha_M}$ & 
$\beta_2$& $\frac{\alpha_S}{2\alpha_M}$ \\ & & &\\ 
\hline & & &
\\
$\beta_3$& $\frac{1}{2}$ & 
$\beta_4$& $\frac{1}{\alpha_S}$ \\& & & \\ 
\hline & & &
\\
$\beta_5$& $\frac{\alpha_M}{2\alpha_S}$ & 
$\beta_6$& $\beta_3$ \\& & & \\ 
\hline 

\end{tabular}
\end{table*}


\begin{corollary}
Following (\ref{eq:average_spectral_efficiency}), the average uplink and downlink spectral efficiencies for the decoupled access modes can be, respectively, given by

\begin{equation}
{\rm{SE}}_{{\rm{UL}}}^D = \sum\limits_{i = 1}^4 {\frac{{{\tau _{{\rm{UL}}{\rm{,Case}}\,i}}}}{W_k}} \Pr ({\rm{Case}}\,i).
\end{equation}
\begin{equation}
{\rm{SE}}_{{\rm{DL}}}^D = \sum\limits_{i = 1}^4 {\frac{{{\tau _{{\rm{DL}}{\rm{,Case}}\,i}}}}{W_k}} \Pr ({\rm{Case}}\,i).
\end{equation}

\end{corollary}

\begin{corollary}
\label{Th:Avg_user_rate_c}

Similar to Theorem \ref{Th:Avg_user_rate_d}, the average user rate for the coupled access mode for the two conventional association cases (i.e., either a UE connects to a $Mcell$ BS or $Scell$ BS ) can be formulated easily following the same steps as adopted in the proof of Theorem \ref{Th:Avg_user_rate_d} in Appendix \ref{appendix_avg_user_rate}. When UEs connect to $Scell$ BSs, in coupled access mode, this association case will be exactly equal to the association case 4 of decoupled access mode. Therefore the average uplink UE rate ${\tau^C _{{\rm{UL}},S}}$ and the  average downlink UE rate ${\tau^C _{{\rm{DL}},S}}$ in coupled access mode, respectively,  are  equal to ${\tau _{{\rm{UL}}{\rm{,Case4}}}}$ and  ${\tau _{{\rm{DL}}{\rm{,Case4}}}}$. When UEs connect to $Mcell$ BSs in coupled access mode, the average user rate of this association case can be formulated by simple mathematical manipulation of expressions of association case 1 and 2, given in Theorem \ref{Th:Avg_user_rate_d}.

\begin{eqnarray}
\label{eq:average_user_rate_M_UL}
{\tau^C _{{\rm{UL}},M}} &=& \frac{{W_M2\pi {\lambda _M}{\beta _1}{\beta _2}{\beta _3}}}{{\Pr (Mcell)}} \cr 
&&\int\limits_{t > 0} {\xi _{{\rm{UL}}{\rm{,}}M}^{ - 2}}{{\rm{\hat H}}\left( {1{\rm{;}}\frac{{{\xi _4}}}{{{\xi _{{\rm{UL}}{\rm{,}}M}}}}{\rm{,}}\frac{{{\xi _5}}}{{{\xi _{{\rm{UL}}{\rm{,}}M}}}}} \right) } dt,
\end{eqnarray}

\begin{eqnarray}
\label{eq:average_user_rate_M_DL}
{\tau^C _{{\rm{DL}},M}} &=& \frac{{W_M2\pi {\lambda _M}{\beta _1}{\beta _2}{\beta _3}}}{{\Pr (Mcell)}} \cr 
&&\int\limits_{t > 0} {\xi _{{\rm{DL}}{\rm{,}}M}^{ - 2}}{{\rm{\hat H}}\left( {1{\rm{;}}\frac{{{\xi _4}}}{{{\xi _{{\rm{DL}}{\rm{,}}M}}}}{\rm{,}}\frac{{{\xi _5}}}{{{\xi _{{\rm{DL}}{\rm{,}}M}}}}} \right) } dt,
\end{eqnarray}


where 
\begin{equation}
\label{eq:pr_mcell}
\begin{split}
\Pr (Mcell) &= \Pr (X_M^{ - {\alpha _M}}{{\bar P}_M} > X_S^{ - {\alpha _S}}{{\bar P}_S}) \\
&= 1 - \frac{{{\alpha _M}}}{{2{\alpha _S}}}H_{1,1}^{1,1}\left[ {{z_2}\left| {\begin{array}{*{20}{c}}
{(0,\frac{1}{2})}\\
{(0,\frac{1}{2}\frac{{{\alpha _M}}}{{{\alpha _S}}})}
\end{array}} \right.} \right]
\end{split}
\end{equation}
is the association probability of a UE to $Mcell$ BS. Rest of the variables used in equations (\ref{eq:average_user_rate_M_UL}) to (\ref{eq:pr_mcell}) are defined in Table \ref{Table:definitions}.
\end{corollary}

\begin{corollary}
Similar to the decoupled access mode, following (\ref{eq:average_spectral_efficiency}),the average uplink and downlink spectral efficiencies for the coupled access modes can be, respectively, given by
\begin{equation}
{\rm{SE}}_{{\rm{UL}}}^C =  \frac{{\tau^C _{{\rm{UL}},M}}}{W_M} \cdot \Pr (Mcell) + \frac{{\tau^C _{{\rm{UL}},S}}}{W_S} \cdot \Pr (Scell),
\end{equation}

\begin{equation}
{\rm{SE}}_{{\rm{DL}}}^C =  \frac{{\tau^C _{{\rm{DL}},M}}}{W_M} \cdot \Pr (Mcell) + \frac{{\tau^C _{{\rm{DL}},S}}}{W_S} \cdot \Pr (Scell),
\end{equation}

where the $\Pr(Scell)=\Pr({\rm{Case}\,4})$.

\end{corollary}

\begin{table}[h]
\centering
\caption{System Parameters}
\label{table:system_paramters}
\begin{tabular}{|l|l|}
\hline
Parameters & Value \\ \hline \hline
$Mcell$ BS transmit power $P_M$ (dBm)           &  46     \\ 
$Scell$ BS transmit power $P_S$ (dBm)           &  20     \\ 
Antenna Gain for $Mcell$ BS $G_M$ (dBi)        &  0      \\ 
Antenna Gain for $Scell$ BS $G_S$ (dBi)        &  18     \\ 
UE's transmit power to $Mcell$ BS $Q_M$ (dBm)   &  20     \\ 
UE's transmit power to $Scell$ BS $Q_S$ (dBm)   &  20     \\ 
Pathloss exponent for sub-6GHz tier $\alpha_M$ &  3     \\ 
Pathloss exponent for millimeter wave tier $\alpha_S$ (LOS) &  2     \\ 
Pathloss exponent for millimeter wave tier $\alpha_S$ (NLOS) &  4     \\ 
Noise power $\sigma^2_S=\sigma^2_M$ (dBm)					   &  0     \\ \hline
\end{tabular}
\end{table}
\section{Numerical and Simulation Results}
\label{sec:num_sim_results}

In this section a comprehensive performance comparison between coupled and decoupled access modes is presented using numerical and simulation results. As presented in section \ref{system_model}, we consider two-tier HetNet, which consists of sub-6Ghz $Mcell$ BSs and millimeter wave $Scell$ BSs are modeled using independent homogeneous PPP. We assume that the tier of $Mcell$ BSs is interference limited whereas the tier of $Scell$ BSs in only noise limited \cite{singh2015tractable,elshaer2016downlink}. The simulation model simply consists of uniformly distributed BSs of both tiers and a UE located in the center of a circular area.
The default system parameters are selected based on the 3GPP specifications {\cite{TR36942}} and existing research work { \cite{smiljkovikj2015analysis,smiljkovikj2015efficiency,elshaer2016downlink}, }their values are listed in Table \ref{table:system_paramters}. {Without any loss of generality the noise power is normalised to 1mW}. To validate the analytical model formulated in this paper, we perform Monte Carlo simulations for all the association cases under consideration.  The simulation results are obtained by averaging over 100,000 independent realizations in MATLAB.  We investigate the potential gain of decoupled access in terms of spectral efficiency and discuss its efficacy from a pragmatic perspective. Moreover, we also discuss the joint association probabilities and distance distributions of a typical UE to its serving BSs for line of sight (LOS) and non line of sight (NLOS) cases. 

\subsection{Joint Association Probabilities}

We first analyze the joint association probabilities of three possible association cases as shown in Fig. \ref{fig:prob_assoc_NLOS} and \ref{fig:prob_assoc_LOS}, for NLOS and LOS scenarios, respectively. {In LOS scenario i.e., when UEs have access to unobstructed link to mmW BSs, UEs choose to connect with $Scell$ BSs in both uplink and downlink with very high probability. The rationale behind this trend is obvious, the antenna gain along with high quality of LOS link of $Scell$ BS becomes an attractive choice with respect to the received power. On the other hand, for NLOS scenario, Fig. \ref{fig:prob_assoc_NLOS} shows that the probability of UEs who choose to decouple are higher only for the lower values of $\lambda_S/\lambda_M$ and it decreases significantly at the cost of increase in association case 4 i.e., when UEs choose to connect with $Scell$ BSs in both uplink and downlink, as we increase the density of $Scell$ BSs. Similarly, in the case when LOS links to $Scell$ BSs are available, which would be the case in less dense urban environments, the probability of UEs who choose to decouple goes to almost zero quite rapidly as we increase the density of BSs $\lambda_S/\lambda_M$ }. 

{
The effect of shadowing is analysed by performing simulations based on the shadowing parameters given in \cite{akdeniz2014millimeter}. In these simulations a lognormal shadowing parameter $\xi  \sim {\mathcal{N}}(0,\sigma _\xi ^2)$ is added in the pathloss. The results demonstrate that the effect of shadowing, as shown in Fig. \ref{fig:prob_assoc_NLOS} and Fig. \ref{fig:prob_assoc_LOS}, does not modify in a measurable way the association curves. Without the loss of generality it can be stated that the effect of shadowing has a minimal impact on the association curves and therefore can be ignored in the rest of the work. 

}
{
Moreover, to understand the impact of pathloss exponents in the association phase, we plotted the probabilities of three possible association cases against different values of the pathloss exponent $\alpha_S$ in Fig. \ref{fig:Pathloss_Exponent_Impact}. It gives us an insight into the effect of BSs' density $\lambda_S/\lambda_M$ on the probability of association case 1 and the impact of $Scell$ BSs link quality on the probability of association case 2 and 4. It is evident from Fig. \ref{fig:Pathloss_Exponent_Impact} that as the link quality of $Scell$ BSs decreases (i.e., higher values of $\alpha_S$), the probability of association case 2 increases at the cost of decrease in association case 4. Furthermore, it also shows that the higher density of $Scell$ BSs makes them  an attractive choice even for the higher values of $\alpha_S$.
}

These results indicate three important things, (i) without any power biasing the number of UEs who chose to decouple are not significant from a system level point of view; (ii) power biasing can certainly be used for load balancing which forces UEs to decouple at the cost of decrease in achievable rate and spectral efficiency; (iii) in a less dense urban environment (i.e., LOS scenario), where load balancing would not be an issue, the feasibility of the implementation of decoupled access is very bleak. 

\begin{figure}[]%
\centering
\subfigure[NLOS]{%
\label{fig:prob_assoc_NLOS}%
\includegraphics[width=\columnwidth]{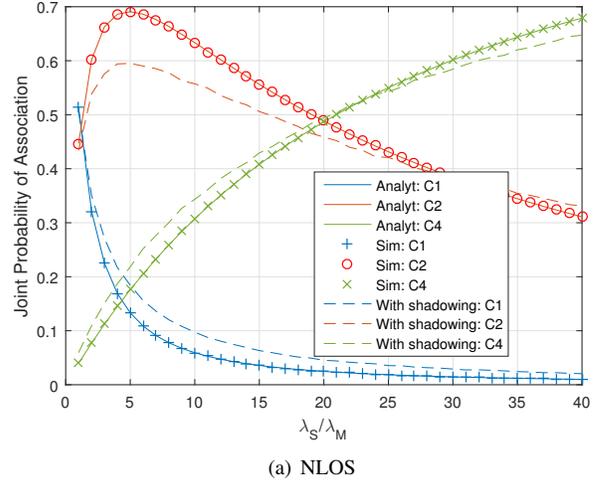}}%
\qquad
\subfigure[LOS]{%
\label{fig:prob_assoc_LOS}%
\includegraphics[width=\columnwidth]{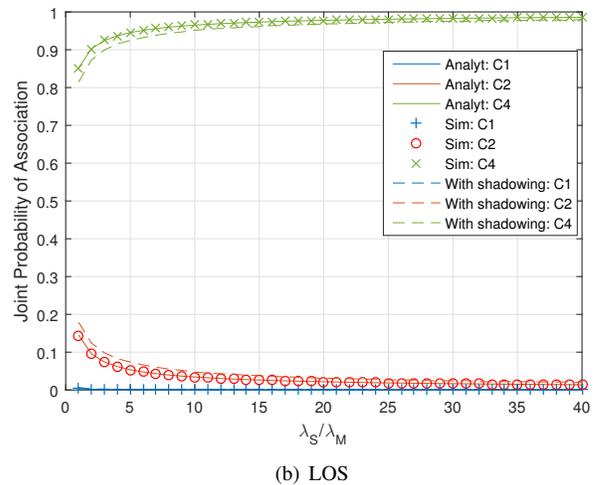}}%
\caption{Joint association probabilities of all three possible association cases.}
\end{figure}

\begin{figure*}[!t]\centering
\includegraphics[width=\columnwidth*2]{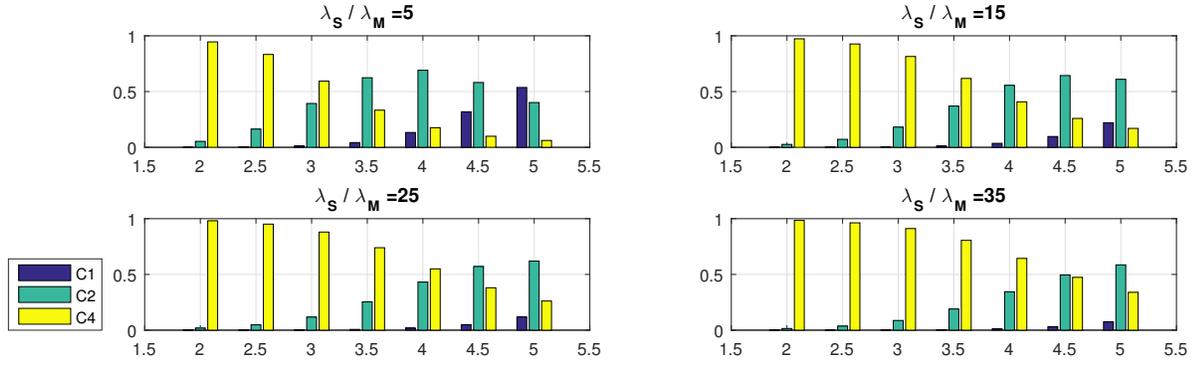}
\caption{{Joint association probabilities of all three possible association cases (y-axis) versus $\alpha_S$ (x-axis) for different densities $\lambda_S/\lambda_M$ of BSs.}}
\label{fig:Pathloss_Exponent_Impact}
\end{figure*}

\begin{figure}[]%
\centering
\subfigure[NLOS]{%
\label{fig:dist_distribution_NLOS}%
\includegraphics[width=\columnwidth]{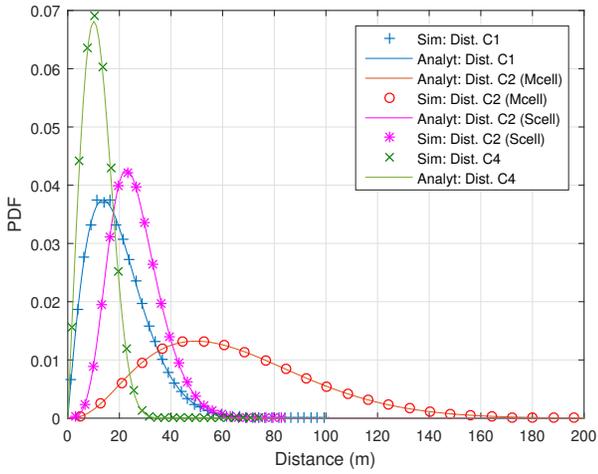}}%
\qquad
\subfigure[LOS]{%
\label{fig:dist_distribution_LOS}%
\includegraphics[width=\columnwidth]{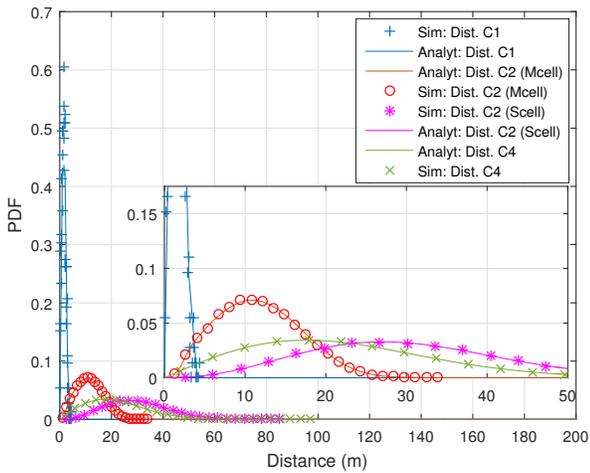}}%
\caption{Distance distributions of a typical UE to its serving BSs.}
\end{figure}

\subsection{Distance Distributions}

Fig. \ref{fig:dist_distribution_NLOS} and \ref{fig:dist_distribution_LOS} shows the distance distributions of a typical UE to its serving BSs for NLOS and LOS cases, respectively. An interesting observation from these  results is that even though the change in the quality of links (i.e., NLOS and LOS)  is only accounted for the millimeter wave or $Scell$ BSs' links, it still affects the distance distribution  of a typical UE to the $Mcell$ BS for the association case 2. It is observed that in NLOS scenario, for the association case 2, the PDF of the distance between the $Mcell$ BS and a typical UE is spread over the range from 0 to 160 meters whereas this range shrinks to 0 to 35 meters for LOS scenario. The rationale behind this change in the distribution is that due to the high quality LOS links of $Scell$ BSs, UEs with very high probability choose to associate with $Scell$ BSs in both uplink and downlink. Therefore, in this scenario decoupling only happens when sub-6GHz or $Mcell$ BSs' link quality is superior to its counterpart and it would only happen if the distance between a typical UE and $Mcell$ BS is significantly less than its counterpart.

\subsection{Spectral Efficiency}

Now we examine the main result i.e., the comparison of spectral efficiency of decoupled and coupled access. It should be obvious that the only difference between the decoupled and coupled access modes are those UEs who choose $Scell$ BSs for uplink in association case 2. Therefore, we only plot and compare the average uplink spectral efficiency of decoupled and coupled access modes. Fig. \ref{fig:Avg_UL_NLOS} compares the average spectral efficiencies of both coupled and decoupled access modes for the NLOS scenario whereas, \ref{fig:Avg_UL_LOS} shows the same comparison for LOS scenario. Even though for the NLOS scenario in Fig. \ref{fig:Avg_UL_NLOS}, a considerable improvement in spectral efficiency can be observed but whether this gain is good enough from pragmatic point of view is still an open question.  

\subsection{Discussion}

There is no doubt that in theory the decoupled access outperforms its coupled counterpart but whether the cost it comes with (i.e., control signals overhead ) makes it viable or not is not a trivial question to answer. From a pragmatic point of view, any proposal for a disruptive change in a system configuration as radical as the decoupled access should be seen from a very critical lens. Therefore following are the key critical insights we can take from this study.

\begin{itemize}
\item Decoupled access does not look a viable option to improve spectral efficiency of an overall system.
\item As previous studies \cite{elshaer2016downlink, zhang2017uplink} suggested that with the help of power biasing, decoupled access would be a viable option for load balancing in the next generation of communication systems. But as we know that the next generation of communication systems are envisioned as part of a dense or an ultra dense network where the density of access points would be greater than the density of UEs\cite{7476821}, hence, whether would there be any need of load balancing in such network settings is a question worth investigating.
\item Decoupled access breaks the channel reciprocity by its very design, therefore, a big challenge is to come up with an cost effective channel estimation scheme. Since at the end it would be the control signals overhead which concretely defines the viability of decoupled access mode. 
\end{itemize}

\begin{figure}%
\centering
\subfigure[NLOS]{
\label{fig:Avg_UL_NLOS}%
\includegraphics[width=\columnwidth]{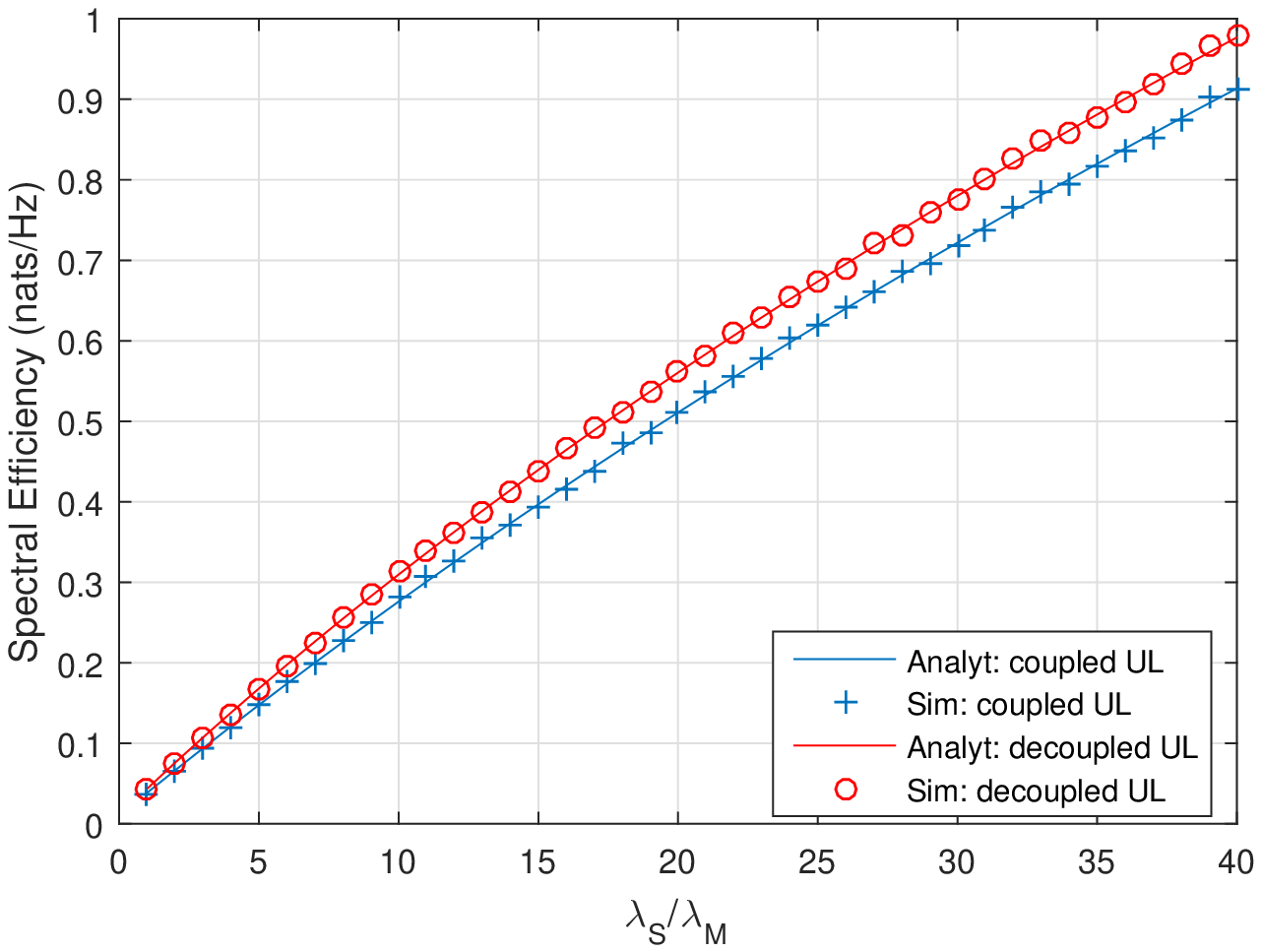}}%
\qquad
\subfigure[LOS]{
\label{fig:Avg_UL_LOS}%
\includegraphics[width=\columnwidth]{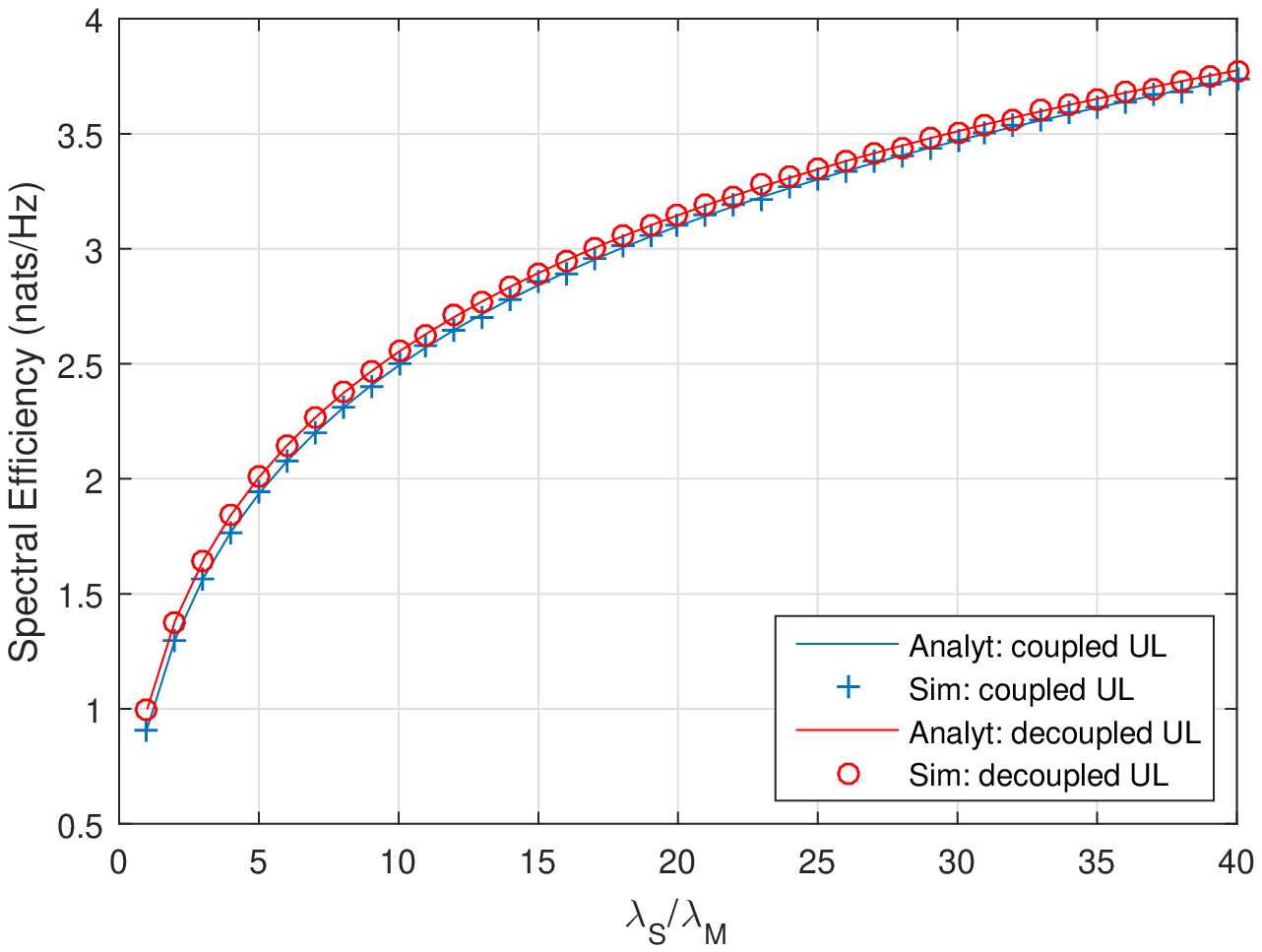}}%
\caption{Comparison of average uplink spectral efficiency of coupled and decoupled access modes.}
\end{figure}

\begin{floatEq}
\begin{eqnarray}
\label{proof_lemma_pr_case_1}
\Pr ({\rm{Case}}1) &=& 2\pi {\lambda _S}\int\limits_0^\infty  \left( {1 - \exp \left( { - \pi {\lambda _M}{{\left( {\frac{{{Q_M}}}{{{Q_S}}}} \right)}^{\frac{2}{{{\alpha _M}}}}}{x^{\frac{{2{\alpha _S}}}{{{\alpha _M}}}}}} \right)} \right) \exp \left( { - \pi {\lambda _S}{x^2}} \right)x\,dx\nonumber \\ 
 &=& 1 - 2\pi {\lambda _S}\int\limits_0^\infty  \left( {\exp \left( { - \pi {\lambda _M}{{\left( {\frac{{{Q_M}}}{{{Q_S}}}} \right)}^{\frac{2}{{{\alpha _M}}}}}{x^{\frac{{2{\alpha _S}}}{{{\alpha _M}}}}}} \right)} \right) \exp \left( { - \pi {\lambda _S}{x^2}} \right)x\,dx.
\end{eqnarray}
\end{floatEq}

\section{Conclusion}
\label{sec:conclusions}
In this paper, we did a comparative analysis of decoupled and coupled access modes based on the stochastic geometry. We constructed a two-tier HetNet (i.e., sub-6Ghz and millimeter wave) where BSs of both tiers are modeled as independent PPP. In our analytical model, we accommodated two different pathloss exponents for two radically different tiers of BSs with respect to their frequency bands.  Therefore, we used Fox's H-function to derive joint probability of associations in closed form, which eventually results in a rather simple, compact and modular expressions for spectral efficiency. Our analytical and simulation results validate each other and illustrate the effect of decoupled access on distance distributions of serving BSs and average spectral efficiency. We observed that though decoupling can improve the uplink spectral efficiency but still that improvement is rather small from a system level point of view. Therefore, whether the decoupled access is a viable option for the next generation of communication systems depends only on a comprehensive cost analysis of control signals overhead.

\appendices
\section{Proof of Lemma 1}
\label{Lemma1_3}
Using (\ref{eq:nearest_distance_BS}) and (\ref{eq:nearest_distance_BS_2}), we can express the probabilistic expression of (\ref{eq:pr_association_case1_1}) in the form shown in (\ref{proof_lemma_pr_case_1}).

Since, the integral given in (\ref{proof_lemma_pr_case_1}) is over the product of two exponentials with different powers of variable $x$, there is no straightforward way to solve this integral. Therefore, we use the theory of Fox's H-function to express this integral in closed form.

\begin{equation} 
\label{proof_lemma_pr_case_1_2}
\begin{split}
\Pr ({\rm{Case}}1)\mathop  = \limits^{(a)} 1 - 2\pi {\lambda _S}\int\limits_0^\infty \left( {\frac{1}{2}\frac{{{\alpha _M}}}{{{\alpha _S}}}}{{\rm{H}}_{0,1}^{1,0}} \left[ {{\xi_1x}\left| {\begin{array}{*{20}{c}}
{( - , - )}\\
{(0,\frac{1}{2}\frac{{{\alpha _M}}}{{{\alpha _S}}})}
\end{array}} \right.} \right] \right.\\ 
\left.\cdot {\frac{1}{2}}{\rm{H}}_{0,1}^{1,0}\left[ {{\xi_2x}\left| {\begin{array}{*{20}{c}}
{( - , - )}\\
{(0,\frac{1}{2})}
\end{array}} \right.} \right]x\right)\,dx \\
\mathop  = \limits^{(b)} 1 - \frac{1}{2}\frac{{{\alpha _M}}}{{{\alpha _S}}}{\rm{H}}_{1,1}^{1,1}\left[ {{z_1}\left| {\begin{array}{*{20}{c}}
{\left( {0,\frac{1}{2}} \right)}\\
{\left( {0,\frac{1}{2}\frac{{{\alpha _M}}}{{{\alpha _S}}}} \right)}
\end{array}} \right.} \right].
\end{split}
\end{equation}

In (\ref{proof_lemma_pr_case_1_2}), (a) is a direct result of   equation (2.9.4) in \cite{kilbas2004h}; (b) follows from Theorem 2.9 in \cite{kilbas2004h}, here, $\xi_1$ = ${\left( {\sqrt {\pi {\lambda _M}} } \right)^{\frac{{{\alpha _M}}}{{{\alpha _S}}}}}{\left( {\frac{{{Q_M}}}{{{Q_S}}}} \right)^{\frac{1}{{{\alpha _S}}}}}$ and $\xi_2$=$\left( {\sqrt {\pi {\lambda _S}} } \right)$.

Similarly, the proofs of lemma \ref{lemma_2} and \ref{lemma_3} can be obtained by following the same steps as adopted here. 
\section{Proof of Lemma 4}
\label{Lemma4_6}
The distance of all the UEs connected to a $Mcell$ BS in the association case 1 satisfies $(X_M^{ - {\alpha _M}} > \frac{{{{\bar Q}_S}}}{{{{\bar Q}_M}}}X_S^{ - {\alpha _S}}) $. Therefore, the complementary CDF of the distance of UEs to their serving BS is formulated as

\begin{equation}
\label{proof_lemma_dist_case_1}
\begin{split}
F_{{X_M}|{\rm{Case}}1}^c = \Pr \left( {{X_M} > x|{X_S} > {{\left( {\frac{{{\bar Q_S}}}{{{\bar Q_M}}}} \right)}^{\frac{1}{{{\alpha _S}}}}}{X_M}^{\frac{{{\alpha _M}}}{{{\alpha _S}}}}} \right)
\\
= \frac{{\Pr \left( {{X_M} > x;{X_S} > {{\left( {\frac{{{\bar Q_S}}}{{{\bar Q_M}}}} \right)}^{\frac{1}{{{\alpha _S}}}}}{X_M}^{\frac{{{\alpha _M}}}{{{\alpha _S}}}}} \right)}}{{\Pr ({\rm{Case}}1)}}
\\
= \frac{{\int\limits_x^\infty  {\exp \left( { - \pi {\lambda _S}{{\left( {\frac{{{\bar Q_S}}}{{{\bar Q_M}}}} \right)}^{\frac{2}{{{\alpha _S}}}}}{x_M}^{\frac{{2{\alpha _M}}}{{{\alpha _S}}}}} \right) \cdot } {f_{{X_M}}}({x_M})d{x_M}}}{{\Pr ({\rm{Case}}1)}}.
\end{split}
\end{equation}

Where ${\Pr ({\rm{Case}}\,1)}$ and ${f_{{X_M}}}$ are given in equations (\ref{eq:pr_association_case1_H}) and (\ref{eq:nearest_distance_BS}), respectively. The cdf of the distance to the serving BS for the association case 1 is $F_{{X_M}|{\rm{Case}}1}=1-F_{{X_M}|{\rm{Case}}1}^c$, and by simply differentiating this cdf, we derive the pdf of the distance to the serving BS for the association case 1.

Similarly, proofs of lemma \ref{lemma:5} and \ref{lemma:6} can be obtained following the same steps. 
\section{Proof of Theorem \ref{Th:Avg_user_rate_d}}
\label{appendix_avg_user_rate}
Following the approach adopted in  \cite{andrews2011tractable}, the average user rate $\tau  \buildrel \Delta \over = W \bE \left[ {\ln (1 + {\rm{SINR}})} \right]$. Here we derive the average user rates, uplink and downlink separately, for the cell association case 2. Same steps can be followed to derive the average user rate expressions for the other two cases where

\begin{eqnarray}
\label{proof_avg_user_rate}
{\tau _{{\rm{UL}}{\rm{,Case2}}}} &\buildrel \Delta \over =& W_S\bE\left[ {{\rm{ln}}\left( {{\rm{1 + SN}}{{\rm{R}}_{{\rm{UL}}{\rm{,}}S}}{\rm{(}}x{\rm{)}}} \right)} \right] \cr
 &\mathop  = \limits^{(c)}&W_S \int\limits_{x > 0}^{} {\int\limits_{t > 0} \left( {\exp \left( -{\frac{{{x^{{\alpha _S}}}\left( {\exp \left( t \right) - 1} \right)}}{{{{\bar Q}_S}}}} \right)} \right. } \cr 
&& \qquad \qquad \qquad \qquad \left.\cdot  {f_{{X_S}|{\rm{Case2}}}}\right)\,dt\,dx,
\end{eqnarray}
(c) comes from following the same steps as given in the proof of Theorem 3 in Appendix C of \cite{andrews2011tractable}. To further solve the integrals of (\ref{proof_avg_user_rate}), we first change the order of integrals. The rationale behind this change is to exploit the properties of Fox's H-function and obtain the result in form of bivariate Fox's H-function. 

\begin{eqnarray}
\label{proof_avg_user_rate_2}
{\tau _{{\rm{UL}}{\rm{,Case2}}}} &=& W_S\int\limits_{t > 0} {\int\limits_{x > 0} {\exp \left( -{\frac{{{x^{{\alpha _S}}}\left( {\exp \left( t \right) - 1} \right)}}{{{{\bar Q}_S}}}} \right)} } \cr 
&& \qquad \qquad \qquad \qquad \qquad \cdot {f_{{X_S}|{\rm{Case2}}}}\,dx\,dt \cr 
&\mathop  = \limits^{(d)}& W_S\int\limits_{t > 0} {\int\limits_{x > 0} {\frac{1}{{{\alpha _S}}}{\rm{H}}_{0,1}^{1,0}\left[ {{\xi_{{\rm{UL}},S}x}\left| {\begin{array}{*{20}{c}}
{( - , - )}\\
{\left( {0,\frac{1}{{{\alpha _S}}}} \right)}
\end{array}} \right.} \right]} }\cr 
&& \qquad \qquad \qquad  \qquad\cdot {f_{{X_S}|{\rm{Case2}}}}\,dx\,dt,
\end{eqnarray}
where ${\xi_{{\rm{UL}},S}} = \frac{{{{\left( {\exp (t) - 1} \right)}^{\frac{1}{{{\alpha _S}}}}}}}{{{{\bar Q}_S}^{\frac{1}{{{\alpha _S}}}}}}$ and (d) follows from eq. (2.9.4) in \cite{kilbas2004h}. Similarly, we can write the exponential terms in the PDF ${f_{{X_S}|{\rm{Case2}}}}$ in the form of Fox's H-function as expressed in (\ref{proof_avg_user_rate_3}). Then, using the result of an integral involving the product of three Fox's H-function provided in \cite{mittal1972integral}, which concludes the proof of ${\tau _{{\rm{UL}}{\rm{,Case2}}}}$. The proof of ${\tau _{{\rm{DL}}{\rm{,Case2}}}}$ also follows the same steps except the fact that now we have to accommodate interference in our expression too. It is assumed that all the communication within each $Mcell$ is based on orthogonal resources, therefore, only inter cell interference is accounted. It is also assumed that each $Mcell$ has at least one active UE which causes interference to its neighboring cells. 
\begin{floatEq}
\begin{eqnarray}
\label{proof_avg_user_rate_3}
{\tau _{{\rm{UL|Case2}}}} &=& W_S\int\limits_{t > 0} {\int\limits_{x > 0} {2\pi {\lambda _S}x\frac{1}{{{\alpha _S}}}{\rm{H}}_{0,1}^{1,0}\left[ {{\xi_{{\rm{UL}},S}x}\left| {\begin{array}{*{20}{c}}
{( - , - )}\\
{\left( {0,\frac{1}{{{\alpha _S}}}} \right)}
\end{array}} \right.} \right]} }  \cdot \left( \frac{1}{2}\frac{{{\alpha _M}}}{{{\alpha _S}}}{\rm{H}}_{0,1}^{1,0}\left[ {{\xi_1x}\left| {\begin{array}{*{20}{c}}
{( - , - )}\\
{\left( {0,\frac{1}{2}\frac{{{\alpha _M}}}{{{\alpha _S}}}} \right)}
\end{array}} \right.} \right] \right. \cr 
&& \qquad \qquad \qquad \left.- \frac{1}{2}\frac{{{\alpha _M}}}{{{\alpha _S}}}{\rm{H}}_{0,1}^{1,0}\left[ {{\xi_3x}\left| {\begin{array}{*{20}{c}}
{( - , - )}\\
{\left( {0,\frac{1}{2}\frac{{{\alpha _M}}}{{{\alpha _S}}}} \right)}
\end{array}} \right.} \right] \right)  \cdot \frac{1}{2}{\rm{H}}_{0,1}^{1,0}\left[ {{\xi_2x}\left| {\begin{array}{*{20}{c}}
{( - , - )}\\
{\left( {0,\frac{1}{2}} \right)}
\end{array}} \right.} \right] \,dx\,dt.
\end{eqnarray}
\end{floatEq}

\begin{eqnarray}
\label{proof_eq_avg_dl_c2_interference}
{\tau _{{\rm{DL}},{\rm{Case2}}}} &\buildrel \Delta \over =& W_M\left[ {{\rm{ln}}\left( {{\rm{1 + SIN}}{{\rm{R}}_{{\rm{DL}},M}}(x)} \right)} \right] \cr
  &\mathop  = \limits^{(e)}&W_M \int\limits_{t > 0} {\int\limits_{x > 0} {\exp \left( { - \frac{{{x^{\alpha_M}}\left( {\exp (t) - 1} \right)}}{{{{\bar P}_M}}}} \right)} } \cr
 && \cdot\bE\left[ {\exp \left( { - \frac{{{x^{\alpha_M}}\left( {\exp (t) - 1} \right){I_{{{\rm{DL}},M}}}}}{{{{\bar P}_M}}}} \right)} \right] \cr 
 && \qquad \qquad \qquad \qquad \cdot {f_{{X_M}|{\rm{Case2}}}} \,dx\,dt,
\end{eqnarray}    
where (e) again follows the similar steps as listed in the proof of Theorem 3 in Appendix C of \cite{andrews2011tractable}. In the (\ref{proof_eq_avg_dl_c2_interference}), the expected value of the interference term can be modeled as a Laplace function, hence it can be formulated as follows

\begin{eqnarray}
\bE\left[ {\exp \left( { - \frac{{{x^{{\alpha _M}}}\left( {\exp (t) - 1} \right){I_{{{\rm{DL}},M}}}}}{{{{\bar P}_M}}}} \right)} \right] \qquad \qquad \qquad \qquad \nonumber
\end{eqnarray}
\begin{eqnarray}
&\mathop  = \limits^{(f)} & \bE\left[ {\exp \left( { - \frac{{{x^{{\alpha _M}}}\left( {\exp (t) - 1} \right)\sum\limits_{v \in {\Phi _{{\rm{IU}}}}} {{I_{{\rm{IU}}}}_{,v}} }}{{{{\bar P}_M}}}} \right)} \right] \cr
&\mathop  = \limits^{(g)}& \bE\left[ {\mathop \prod \limits_{v \in {\Phi _{{\rm{IU}}}}} \exp \left( { - \frac{{{x^{{\alpha _M}}}\left( {\exp (t) - 1} \right){I_{{\rm{IU}}}}_{,v}}}{{{{\bar P}_M}}}} \right)} \right]\cr
&\mathop  = \limits^{(h)}& \exp \bigg(  - 2\pi {\lambda _{{\rm{IU}}}} 
\cr 
&&\left. \cdot \int\limits_{y > 0} {\left( {1 - \frac{1}{{1 + \left( {\frac{{{x^{{\alpha _M}}}\left( {\exp (t) - 1} \right)}}{{{{\bar P}_M}}}} \right){{\bar P}_M}{y^{ - {\alpha _M}}}}}} \right)y\,dy}  \right)\cr 
&\mathop  = \limits^{(i)} &\exp \left(  - \pi {\lambda _{{\rm{IU}}}}{x^2}{{\left( {\exp (t) - 1} \right)}^{\frac{2}{{{\alpha ^M}}}}}\frac{2}{{{\alpha ^M}}}  \right. \cr 
&& \qquad \qquad \qquad \qquad \qquad  \left. \cdot\int\limits_{{{\left( {\exp (t) - 1} \right)}^{ - 1}}}^\infty  {\frac{{{u^{\frac{2}{{{\alpha ^M}}} - 1}}}}{{1 + u}}} du \right) \cr 
&\mathop  = \limits^{(j)} &\exp \left( { - \pi {\lambda _{{\rm{IU}}}}{x^2}G(t)} \right),
\end{eqnarray}
where (f) is simply the expectation over the distance between a typical UE and its interferers; (g) follows from the simple fact that channel between the interferers  and a typical UE is i.i.d. and it is independent from the point process of interferers $\phi_{\rm{IU}}$; (h) follows from the probability generating functional \cite{chiu2013stochastic} of the PPP, which states that $\bE\left[ {\mathop {\prod f(x)}\limits_{x \in {\Phi _{}}} } \right] = \exp \left( { - \lambda \int_{{{\rm I\!R}^2}} {\left( {1 - f(x)} \right)dx} } \right)$; (i) is a result of some trivial mathematical manipulation and change of variable. In (j), we solve the inner integral over $u$ by using the eq. given in section 3.194 of \cite{mittal1972integral}, where the  $G(t)$ is defined in (\ref{eq:Gt}).
\begin{floatEq}
\begin{eqnarray}
\label{eq:Gt}
G(t) = {\left( {\exp (t) - 1} \right)^{\frac{2}{{{\alpha ^M}}}}}\frac{2}{{{\alpha ^M}}}\left( {\frac{{{{\left( {\exp (t) - 1} \right)}^{1 - \frac{2}{{{\alpha ^M}}}}}}}{{1 - \frac{2}{{{\alpha ^M}}}}}{}_2{F_1}\left[ {1,1 - \frac{2}{{{\alpha ^M}}};2 - \frac{2}{{{\alpha ^M}}};\frac{{ - 1}}{{{{\left( {\exp (t) - 1} \right)}^{ - 1}}}}} \right]} \right)
\end{eqnarray}
\end{floatEq}
Similar to the steps adopted in the proof of ${\tau _{{\rm{UL}}{\rm{,Case2}}}}$, we can replace the exponential terms in the expression of ${\tau _{{\rm{DL}},{\rm{Case2}}}}$ with their respective Fox's H-function as formulated in  (\ref{proof_avg_user_rate_4}). Then, again using the result of an integral involving the product of three Fox's H-function provided in \cite{mittal1972integral}, we concludes the proof of ${\tau _{{\rm{DL}}{\rm{,Case2}}}}$.

\begin{floatEq}
\begin{eqnarray}
\label{proof_avg_user_rate_4}
{\tau _{{\rm{DL|Case2}}}} &=& W_M\int\limits_{t > 0} {\int\limits_{x > 0} {2\pi {\lambda _M}x\frac{1}{{{\alpha _M}}}{\rm{H}}_{0,1}^{1,0}\left[ {{\xi_{{\rm{DL}},M}x}\left| {\begin{array}{*{20}{c}}
{( - , - )}\\
{\left( {0,\frac{1}{{{\alpha _M}}}} \right)}
\end{array}} \right.} \right]} }  \cdot \left( \frac{1}{2}\frac{{{\alpha _S}}}{{{\alpha _M}}}{\rm{H}}_{0,1}^{1,0}\left[ {{\xi_4x}\left| {\begin{array}{*{20}{c}}
{( - , - )}\\
{\left( {0,\frac{1}{2}\frac{{{\alpha _S}}}{{{\alpha _M}}}} \right)}
\end{array}} \right.} \right] \right. \cr 
&& \qquad \qquad \qquad \qquad \left. - \frac{1}{2}\frac{{{\alpha _S}}}{{{\alpha _M}}}{\rm{H}}_{0,1}^{1,0}\left[ {{\xi_6x}\left| {\begin{array}{*{20}{c}}
{( - , - )}\\
{\left( {0,\frac{1}{2}\frac{{{\alpha _S}}}{{{\alpha _M}}}} \right)}
\end{array}} \right.} \right]\right)  \cdot \frac{1}{2}{\rm{H}}_{0,1}^{1,0}\left[ {{\xi_5x}\left| {\begin{array}{*{20}{c}}
{( - , - )}\\
{\left( {0,\frac{1}{2}} \right)}
\end{array}} \right.} \right] \,dx\,dt.
\end{eqnarray}
\end{floatEq}

\section{Structure of bivariate Fox's H-function }
\label{Appendix:definitions}
In this section of the paper, the structure of the bivariate Fox's H-function is defined in (\ref{eq:defnition_H}), shown on the next page.

\begin{floatEq}
\begin{eqnarray}
\label{eq:defnition_H}
\hat{\rm{H}}(k;x,y) &=&\rm{H}\left[ {\begin{array}{*{20}{c}}
{\left( {\begin{array}{*{20}{c}}
0&1\\
1&0
\end{array}} \right)}\\
{\left( {\begin{array}{*{20}{c}}
1&0\\
0&1
\end{array}} \right)}\\
{\left( {\begin{array}{*{20}{c}}
1&0\\
0&1
\end{array}} \right)}
\end{array}\left| {\begin{array}{*{20}{c}}
{\left( {\begin{array}{*{20}{c}}
{1 - 2{\beta _k};{\beta _k},{\beta _k}}\\
{ - ; - , - }
\end{array}} \right)}\\
{\left( {\begin{array}{*{20}{c}}
{ -  - }\\
{0,{\beta _{k + 1}}}
\end{array}} \right)}\\
{\left( {\begin{array}{*{20}{c}}
{ -  - }\\
{0,{\beta _{k + 2}}}
\end{array}} \right)}
\end{array}} \right.\left| {\begin{array}{*{20}{c}}
{\begin{array}{*{20}{c}}
{}\\
{}
\end{array}}\\
{\begin{array}{*{20}{c}}
{\left( {x,y} \right)}\\
{}
\end{array}}\\
{\begin{array}{*{20}{c}}
{}\\
{}
\end{array}}\nonumber
\end{array}} \right.} \right] \\
\end{eqnarray}

\end{floatEq}

\ifCLASSOPTIONcaptionsoff
  \newpage
\fi



%




\bibliographystyle{bibtex/bib/IEEEtran}
\bibliography{bibtex/bib/IEEEabrv,bibtex/bib/references}{}

\end{document}